\documentclass[times, 10pt,onecolumn]{article}

\usepackage{times}
\usepackage{epsf}
\usepackage{latexsym}
\usepackage{amsmath}
\usepackage{amsfonts}
\usepackage{euscript}
\usepackage{url}

\font\sixteen=cmr10 at 16pt

\newcommand{\ignore}[1]{}

\newcommand{\isequal}{\stackrel{\rm ?}{=}}

\newcommand{\qed}{\hspace*{\fill}\rule{7pt}{7pt}}
\newenvironment{proof}{\noindent{\bf Proof.}\hspace*{1em}}{\qed\medskip}
\newtheorem{definition}{Definition}
\newtheorem{theorem}{Theorem}

\begin{document}

\title{Efficient Defence against Misbehaving TCP Receiver DoS Attacks}
\author{Aldar C-F. Chan\\
Institute for Infocomm Research, A*STAR Singapore}

\maketitle

\begin{abstract}
The congestion control algorithm of TCP relies on {\em correct} feedback from the receiver to determine the rate at which packets should be sent into the network.  Hence, correct receiver feedback (in the form of TCP acknowledgements) is essential to the goal of sharing the scarce bandwidth resources fairly and avoiding congestion collapse in the Internet.  However, the assumption that a TCP receiver can always be trusted (to generate feedback correctly) no longer holds as there are plenty of incentives for a receiver to deviate from the protocol.  In fact, it has been shown that a misbehaving receiver (whose aim is to bring about congestion collapse) can easily generate acknowledgements to conceal packet loss, so as to drive a number of honest, innocent senders arbitrarily fast to create a significant number of non-responsive packet flows, leading to denial of service to other Internet users.  We give the first formal treatment to this problem.  We also give an efficient, provably secure mechanism to force a receiver to generate feedback correctly; any incorrect acknowledgement will be detected at the sender and cheating TCP receivers would be identified.  The idea is as follows: for each packet sent, the sender generates a tag using a secret key (known to himself only); the receiver could generate a proof using the packet and the tag alone, and send it to the sender; the sender can then verify the proof using the secret key; an incorrect proof would indicate a cheating receiver.  The scheme is very efficient in the sense that the TCP sender does not need to store the packet or the tag, and the proofs for multiple packets can be aggregated at the receiver.  The scheme is based on an aggregate authenticator.  In addition, the proposed solution can be applied to network-layer rate-limiting architectures requiring correct feedback.

\end{abstract}

\section{Introduction}
The congestion control algorithm \cite{Jacobson88} of TCP (Transmission Control Protocol \cite{rfc793}) is an essential component contributing to the stability of the current Internet \cite{KMT98}. It allows fair sharing of the scarce bandwidth resources \cite{CR89} and avoids congestion collapse in the Internet \cite{KMT98}.  However, it has been found vulnerable to incorrect feedback from a misbehaving receiver by Savage et. al. \cite{SCWA99} and subsequently by Sherwood et. al. \cite{SBB05}.  More specifically, when a misbehaving TCP receiver generates optimistic acknowledgements (acknowledging the receipt of TCP segments which have not been received), the sender of the TCP session could be deceived into believing a fast connection is available and, as a result, injects more packets into the network (through increasing its TCP flow window and sending new packets) despite that the acknowledged TCP segments may have been lost in a heavily congested network.  Note that the stability of the Internet relies on TCP senders to back off their traffic when congestion occurs as inferred by packet loss \cite{Jacobson88,FF99}.  Such a vulnerability in the TCP protocol itself could be threatening to the resilience of the Internet infrastructure as the effort needed for an attacker to bring down part of the Internet is minimal.

The TCP congestion control algorithm \cite{Jacobson88,rfc793,rfc2581} was designed based on the assumption that a TCP receiver can always be trusted \cite{clark88} and all his feedback in the form of acknowledgements \cite{Allman98} for correctly received packets is a reliable source of information about the traffic/congestion condition in the network.  When network congestion occurs at some routers in the network, packet dropping starts to occur at the input queues of the congested routers.  The design of the TCP congestion control algorithm considers packet loss as the onset of congestion somewhere in the network  \cite{Jacobson88,rfc793,rfc2581,FF99}.  The TCP sender thus slows down its rate of injecting packets into the network (through multiplicatively reducing its TCP flow window and slightly inflating its timeout timer) with a view to relieving the congestion condition.  However, the information of packet loss is usually provided solely by the receiver in the form of TCP acknowledgements \cite{Allman98}.  When this information is false, the TCP sender may respond to a network congestion in the oppositive direction, leading to non-responsive traffic flows which cause other {\em honest} TCP users to back off their traffic.  These non-responsive flows, in essence, form a type of denial of service (DoS) attacks, as first identified by Floyd et. al. \cite{FF99} and subsequently illustrated by Sherwood et. al. \cite{SBB05}.  The effect of such non-responsive flows can easily be multiplied by a deliberate attacker using a small TCP segment size \cite{Reed01}.  It should be emphasized that the innocent TCP sender is unaware of the fact that his traffic has caused a serious congestion; the sender is deceived by the malicious receiver.  More seriously, subtle attackers can exploit this vulnerability of TCP to bring down the operation of the Internet by pumping in an enormous amount of traffic through innocent TCP senders deceived by a distributed network of malicious TCP receivers.  In fact, Sherwood et. al. have shown the possibility of bringing terribly extensive network congestion over a prolonged period using such an attack \cite{SBB05}.

Unlike most distributed denial of service (DDoS) attacks, the attacker in concern does not manipulate machines under his control to send traffic to the victim.  Instead, the machines controlled by the attacker would solicit traffic from a possibly large number of innocent victims, which are usually trustworthy (say, a reputable website hosted by the national government).  These traffic streams would converge towards the attack machines and cause serious congestion in the network as they merge together.  There is an additional level of indirection, compared to usual DDoS attacks \cite{Paxson01}.
Usually, DoS defence mechanisms assume a malicious sender and some of them even rely on the receiver to trigger the start of, say, a rate-limiting mechanism \cite{MBFIPS02,ABFKMS08,AC09}.  Hence, these mechanisms would be ineffective in identifying a misbehaving TCP receiver or limiting the impact of its DoS attack.
Puzzle-based defence \cite{WR04} could possibly limit the rate at which a single server/sender injects traffic into the network; but this still cannot restrict a malicious TCP receiver from attracting traffic.  Note that a TCP receiver can acknowledge a whole window of TCP segments/packets using just a single TCP acknowledgement.  A clever attacker can open many TCP sessions with different servers, requesting large files; while each of these servers may only send data at moderate rates, all these flows converge towards the attacker and the resulting traffic aggregate (from all the servers) could be large and overwhelm several en-route routers close to the attacker.  The non-responsive nature of these flows can make the congestion situation equally bad as in usual DDoS attacks.  Of course, ad hoc approaches using a firewall which denies all incoming traffic to certain destinations, as used in the Estonian defence \cite{EstoniaAttack07}, could effectively protect the local network.  But congestion could still happen for a considerable period at various routers in the Internet core.  Moreover, this reactive approach is passive; preventive approaches are desirable and essential for the stability of the Internet, in particular, as the Internet is increasingly complex and TCP is its core protocol at present and likely also in the future.

While the misbehaving TCP receiver attack may not be a common DoS attack due to the sophisticated technical knowledge needed in launching such an attack successfully, it is undoubted that the destruction caused by a successful attack could be huge in the Internet fabric \cite{SBB05}.  Although the occurrence probability is uncertain or might not be very large, considering the probability-impact product of the misbehaving TCP receiver attack, it is still worthwhile to seek a fix for such a known vulnerability.  In particular, this attack could be a very significant security threat to the availability of certain time-critical networks, such as the envisioned Smart Grid using the converged network approach.

Despite the serious impact which could be brought about by an attacker-manipulated set of misbehaving TCP receivers (whose goal is not merely to obtain faster downloads but instead to exploit this TCP vulnerability to mount a distributed denial of service attack against the entire Internet), there lacks a rigorous study and modeling of the security goal against malicious TCP receivers.  As a consequence, it is fair to say the effectiveness of the existing solutions \cite{SCWA99,SBB05} remains uncertain.
On the contrary, the solution proposed in this paper is provably secure \cite{GM84,KM07} in the sense that any breach of security found in the proposed solution can be traced back to a vulnerability in the underlying cryptographic primitive (which is widely assumed to be secure), meaning the insecurity of the latter is the only cause for the insecurity of the former.  However, our solution requires modifications on both TCP senders and receivers.  Nevertheless, no modification to the core functionality of the congestion control algorithm of TCP is made in our solution; changes (in our solution) will only be implemented as a TCP option and incorporated into codes for handling TCP acknowledgements and session management specified in RFC 793 \cite{rfc793}.  Neither will modifications be necessary to any routers.

We address the security problem caused by a misbehaving TCP receiver using an approach which forces the TCP receiver to prove the receipt of TCP segments based on cryptographic tools.  In essence, we force the TCP receiver to tell the truth.  We call the protocol a Verifiable Segment Receipt (VSR) protocol.  In the VSR protocol, a TCP receiver has to generate some proof to convince the sender that all the packets he claims to have received are actually received.  The proof is constructed from the received segments and the attached tags generated by the sender with his secret key.  Without the sender's secret key, it is difficult to generate a valid tag for a particular session.  When a misbehaving TCP receiver wants to lie about the receipt of a certain TCP segment he has actually not received, he will not be able to create (without the knowledge of the lost segment and its tag) a valid VSR proof (acceptable by the sender) even though he may possess the data being transferred (through a prior download)\footnote{An attacker may have knowledge of the payload of the concerned TCP segment through prior downloads from the same sender/server.}; in other words, a malicious TCP receiver cannot launch a playback attack, namely, he cannot generate a valid VSR proof for a lost TCP segment in the current session even though he has complete knowledge of the payload of that segment.  All forged proofs (for segments not received) will be caught by the sender which can then terminate the associated TCP connection.  We propose a VSR protocol construction based on an (additive) aggregate authenticator \cite{CCMT08}.  The protocol is provably secure assuming the indistinguishability property of a pseudorandom function \cite{GGM86}.  A merit of our VSR protocol construction is that the validity of a proof from the TCP receiver can be verified using the sender's secret key only, and {\em the sender does not need to remember the content of any TCP segment or its tag even though the corresponding VSR proof has not yet been received}.  In addition, the VSR proofs for multiple TCP segments can be aggregated to minimize the bandwidth overhead.

It should be emphasized that the VSR protocol addresses a problem different from that of remote integrity check \cite{CX09} or proof of retrievability \cite{JK07,ABCHKPS07}.  There are two subtle differences.  First, the VSR protocol has to guard against playback attacks in which the attacker possesses the data being transferred; while a prover in possession of the stored file can definitely produce a valid proof in the proof of retrievability scheme.  Second, a receiver only needs to generate a proof once (in the current session) for each piece of data received in the VSR protocol\footnote{In the VSR protocol, a new instantiation will be executed for a new session even for the same piece of data.} whereas several executions of the verification protocol for the same piece of data in a proof of retrievability scheme are expected.

\medskip
\noindent {\bf Contributions.}  The contributions of this paper are two-fold: (1) A formal treatment to the security problem of malicious TCP receivers is given, which is a key step for achieving provable security \cite{GM84,KM07}. This paper gives the first formal security model in the literature for the misbehaving TCP receiver problem.  (2) A provably-secure protocol to solve the problem is proposed and the proposed construction is the first provably secure solution for the misbehaving TCP receiver problem.  Besides, the protocol is reasonably efficient to be incorporated into the TCP protocol; empirical studies are also done to show that the proposed solution is practical in terms of computation, communication and storage overhead.  Furthermore, the TCP sender in the VSR protocol does not need to store any TCP segment or its tag in order to verify the validity of the receiver's VSR proof for that segment.  In other words, no storage is required at the TCP sender or receiver.

While the problem considered in this paper is mainly in the context of network security, we give a cryptographic treatment to this problem, including the formal security model and cryptographic solution.  The efficiency of the proposed solution could possibly weaken the folklore impression in the network community that cryptography can seldom be applied to network security in a light-weight manner.

\bigskip
The organization of the paper is as follows.  A formal security model, along with possible security notions for the VSR protocol, is given in Section \ref{sect:security-model}.  Section \ref{sect:construction} presents the VSR protocol.  The computation and communication overhead of the proposed protocol is discussed in Section \ref{sect:compute}, along with other practical considerations such as incremental deployment.  Section \ref{sect:related-work} discusses related work in the literature.

\section{\label{sect:security-model} Security Model for Verifiable Segment Receipt (VSR) Protocol}

Recall that our defence against a misbehaving TCP receiver is based on the idea requiring a TCP receiver to prove to a sender that it has received all the TCP segments it claims.  We call this protocol a Verifiable Segment Receipt (VSR) protocol.  More specifically, the TCP sender attaches some information as tags to segments sent, whereby the TCP receiver can construct a proof from the segments and the corresponding tags received; the proof is then returned as part of the TCP acknowledgement for the sender to verify.  In a properly designed scheme, a misbehaving TCP receiver (which claims the receipt of TCP segments it actually has not received) can only pass the sender's verification with negligibly small probability.

\subsection{Syntax Definitions}
A VSR protocol consists of one sender $S$ and one receiver $R$.  Suppose the victim (sender $S$) has a universal set ${\cal F}$ of files, each with a unique file identifier $fid \in \{0,1\}^*$.  Without loss of generality, $fid$ can be considered as the filename of a particular file in the set ${\cal F}$ held by $S$.  Let $M \in \{0,1\}^*$ be the content of the file to be transferred from $S$ to $R$.  Suppose $S$ and $R$ have agreed on a default segment size in the initial TCP handshake of a usual TCP connection (as specified in RFC 793 \cite{rfc793}) such that the file content $M$ is divided to fill $n$ equally-sized TCP segments $s_i \in \{ 0,1 \}^*$ indexed by $i$ where $1 \le i \le n$.

Each $s_i$ consists of a $40$-byte TCP header $h_i$ and an $l$-byte message payload $m_i$, that is, $s_i = h_i ||  m_i$ (where $||$ denotes concatenation).  Note that the sequence number of the TCP segment $s_i$ is $(i-1) \cdot l + ISN + offset$ where $ISN$ is the initial sequence number of a TCP connection and $offset$ accounts for the bytes used in connection initiation.  What we actually demonstrate is that segment indices (used in this paper) and sequence numbers in a normal TCP connection are interchangeably convertible to each other; we will stick to using segment indices in the discussion of this paper.  Without loss of generality, assume $M = m_1|| \cdots || m_i || \cdots || m_n$.\footnote{In case $|M|$ is not an integral multiple of $|m_i|$, $m_n$ would be padded with zeros.}   A VSR protocol consists of the following algorithms.

\begin{description}
\item[Setup ($\textsf{setup}$).]
$\textsf{setup}(1^{\lambda}) \rightarrow param$ takes the security parameter $\lambda$ and outputs the public parameters $param$ (such as the hash function and its size) used in the VSR protocol.

\item[Key Generation ($\textsf{KG}$).]
Let $\textsf{KG}(1^{\lambda}, param, n) \rightarrow sk$ be a {\em probabilistic} algorithm which generates the secret key $sk$ to be used for a specific TCP session.

\item[Tag Generation ($\textsf{TG}$).]
$\textsf{TG}(param, sk, i, s_i) \rightarrow t_i$ is a {\em probabilistic} algorithm taking a TCP segment $s_i$ (with its index $i$) and the sender secret key $sk$ to generate a tag $t_i$ appended to $s_i$. That is, $(s_i, t_i)$ is sent to the receiver $R$ and the actual TCP segment sent is of the form $s_i || t_i$.


\item[Proof Generation ($\textsf{PG}$).]
 Let $I \subseteq [1,n]$.  Denote $\{s_i: i \in I \}$ by ${\cal S}(I)$ and $\{t_i: i \in I\}$ by ${\cal T}(I)$.  Then $\textsf{PG}(param, I, {\cal S}(I), {\cal T}(I)) \rightarrow p(I)$ is a {\em deterministic} algorithm used by the receiver $R$ to generate a proof $p(I)$ to convince the sender $S$ that all TCP segments in ${\cal S}(I)$ have been received.  $(I, p(I))$ is returned along with the TCP acknowledgement to $S$.

 $\textsf{PG}$ and $p(I)$ are so defined based on $I$ for the sake of generality such that both normal TCP (with cumulative acknowledgements \cite{rfc793}) and TCP-SACK (selective acknowledgement option \cite{rfc2018}) are covered.  For TCP with cumulative acknowledgements \cite{rfc793}, each $I$ is of the form $\{1, 2, ..., r \}$ where $r$ is the index of the last in-sequence segment received correctly.  In actual implementations, only $r$, rather than the whole set $I$, needs to be returned to the sender for TCP with cumulative acknowledgements.

\item[Proof Verification ($\textsf{PV}$).]
$\textsf{PV} (param, I, {\cal T}(I), p(I)) \rightarrow 0/1$ is a {\em deterministic} algorithm to check the correctness of the proof $p(I)$.  If correct, `$1$' is returned, meaning all segments in $I$ have been received by the prover/receiver $R$ which generates a correct proof $p(I)$.  `$0$' would indicate that the receiver lies.

Note that the set of tags ${\cal T}(I)$ may not be needed in some constructions; for instance, our VSR protocol construction in Section \ref{sect:construction} does not need tags for proof verification.  It is included in the syntax to make the model general.

\end{description}

{\em Since $param$ is public and forms part of the input to all algorithms in a VSR protocol, we will not explicitly write it as algorithm input in this paper.}

\subsection{Security Notions}

Four types of oracle queries (adversary interaction with the system) are allowed in the VSR security model, namely, the playback oracle ${\cal O}_{PB}$, the segment query oracle ${\cal O}_{SQ}$, the tag access oracle ${\cal O}_{TA}$ and the proof verification oracle ${\cal O}_{PV}$.  Their details
are as follows:

\begin{description}
\item[Playback Oracle ${\cal O}_{PB}(fid, l)$.]
On input a playback oracle query $\langle fid, l \rangle$, the playback oracle retrieves the file with identifier $fid$ and content $M \in \{0,1\}^*$ and breaks down $M$ into $l$-byte units to fill a number of (say $n$) segments $s_i$ where $1 \le i \le n$; the TCP headers are determined by emulating the initiation of a new TCP session.  The playback oracle then generates the secret key $sk$ by running $\textsf{KG}$ and computes tags $t_i = \textsf{TG}(sk, i, s_i)$ for all $i \in [1, n]$, and replies with $\{(s_i, t_i): 1 \le i \le n \}$.

\item[Segment Query Oracle ${\cal O}_{SQ}(i)$.]
For a current and established TCP session (with fixed default segment size $l$ and fixed total number $n$ of segments for a file with identifier $fid$ and content $M$), for fixed secret key and segment decomposition, on input a segment query $\langle i \rangle$ (where $1 \le i \le n$), the segment query oracle retrieves the segment $s_i$ and the secret key $sk$ and runs the tag generation algorithm $\textsf{TG}$ to generate the tag $t_i = \textsf{TG}(sk, i, s_i)$ and sends $(s_i, t_i)$ as the reply.

\item[Tag Access Oracle ${\cal O}_{TA}(i)$.]
For a current and established TCP session (with fixed default segment size $l$ and fixed total number $n$ of segments for a file with identifier $fid$ and content $M$), for fixed secret key and segment decomposition, on input a tag access oracle query $\langle i \rangle$ (where $1 \le i \le n$), the tag access oracle retrieves the segment $s_i$ and the secret key $sk$ and runs the tag generation algorithm $\textsf{TG}$ to generate the tag and sends the result $\textsf{TG}(sk, i, s_i)$ as the reply.

\item[Proof Verification Oracle ${\cal O}_{PV}(I, p(I))$.]
For a current and established TCP session (with fixed default segment size $l$ and fixed total number $n$ of segments for a file with identifier $fid$ and content $M$), for fixed secret key and segment decomposition, on input a proof verification oracle query $\langle I, p(I) \rangle$ (where $I \subseteq [1,n]$), the proof verification oracle runs $\textsf{PV}$ and relays the result ($0$ or $1$) back.

\end{description}
Normally, the playback oracle is accessed by an adversary before launching the actual attack session while the segment query, tag access and proof verification oracles are usually accessed in the actual attack session.  The playback oracle is to model the scenario that an adversary first accesses the victim websites (by honestly following the TCP congestion control algorithm and sending back correct feedback) to learn their behavior and obtain a number of large files (along with the used TCP headers and VSR tags) and stores this information as reference for launching the subsequent actual attack.  The playback attack would be a real threat in the misbehaving TCP receiver attack.

The difference between the segment query oracle and the tag access oracle is in the reply content: an adversary can only obtain a tag along with the corresponding segment in a segment query, whereas, an adversary can obtain (via some unknown means such as a man-in-the-middle attack) just the tag (without a segment) of a lost TCP segment in a tag access query.

Note that the tag access oracle can be treated as some form of  man-in-the-middle attacks in the narrow sense.  The middle man in a man-in-the-middle attack could perform various computation on a set of captured TCP segments and their VSR tags and pass the result of computation to the attacker while a tag access oracle only detaches the tag from a captured TCP segment and passes the tag to the attacker.

\smallskip
As a basic requirement, a VSR protocol needs to satisfy the correctness property stated as follows.

\smallskip
\noindent {\bf Correctness.}
For all files with content $M \in \{0,1\}^*$, all choices of secret key $sk$, and all $I \subseteq [0,n]$, if $s_i$ is a valid and correctly indexed segment from $M$ and $t_i \leftarrow \textsf{TG}(sk, i, s_i)$ for all $i \in I$, then
\[\textsf{PV} (I, {\cal T}(I), \textsf{PG}(I, {\cal S}(I), {\cal T}(I))) = 1.\]

\smallskip
Recall that the goal of a VSR protocol is to ensure that a TCP receiver cannot deceive a TCP sender into believing that it has received a TCP segment which is actually lost in the network.  In the usual networking environment, a TCP receiver would receive most of the transmitted TCP segments (and their VSR tags) while a small portion of them may be dropped at the input queues of congested routers enroute. In the normal situation, it is implicitly assumed that the tags for all lost segments will not be accessible to a TCP receiver; this assumption could be reasonable since when a TCP segment (more precisely, an IP datagram) is dropped at a congested router, there is no way for a TCP receiver to learn the tag of that TCP segment which has never arrived at the receiver unless there is some out-of-band channel accessible to the receiver.\footnote{For instance, the misbehaving TCP receiver installs some agents close to the sender to grab the tag of a TCP segment before it is dropped by a router and to send him just the tag.  However, this channel is contrive.}  Nevertheless, for completeness, in our model, we include a tag access oracle to provide tags of lost segments for an adversary.

Informally, there are two types of attacks for a VSR protocol to withhold: known file content (KFC) attacks and tag access (TA) attacks.  In a known file content attack, the adversary has access to the file content that a TCP session transfers; the adversary could have downloaded the whole file content beforehand through honestly running another TCP session; when launching the actual attack session, he can leverage on this downloaded content; consequently, a proper VSR protocol construction should add some sort of freshness into the actual file content to guard against the known file content attack.  In a tag access attack, the adversary can have some unknown means (say man-in-the-middle attacks) to access the tag of a TCP segment without obtaining the whole TCP segment.

We have three notions for the security (or soundness) of the VSR protocol.  The first one (security against known file content and tag access attacks) being the strongest notion might not be achievable.  The second one (security against known file content only attacks) could be the default one which reasonably fits real world scenarios.  The third one (security against tag access attacks) could be considered as a weak form of man-in-the-middle attacks in the context of this paper.

\subsubsection{Security against Known File Content and Tag Access Attacks ($\textsf{KFC-TA}$-secure).}

To define security against know file content and tag access attacks ($\textsf{KFC-TA}$ attacks) for the VSR protocol, we use the following game played between a challenger and an adversary.  If no PPT (Probabilistic, Polynomial Time) adversary can win the game with non-negligible advantage (as defined below), we say the VSR protocol is $\textsf{KFC-TA}$-secure.

\begin{definition}
\label{def:kfc-ta-secure}

A VSR protocol is secure against known file content and tag access attacks (i.e. $\textsf{KFC-TA}$-secure) if the advantage of winning the following
game is negligible in the security parameter $\lambda$ for all $PPT$
adversaries.
\end{definition}
\begin{description}
\item[Setup.]
The challenger runs $\textsf{setup}$ to generate the public parameters $param$ and gives $param$ to the adversary.

\item[Query 1.]
The adversary can issue to the challenger any playback queries on file $fid_j$ and default segment size $l_j$ of his choice.  Assuming the file content $M_j$ is broken down to fill $n_j$ segments, the challenger picks a new secret key $sk_j$ (if needed in the protocol) by running $\textsf{KG}$ and responds with $\{(s_{ji}, \textsf{TG}(sk_j, i, s_{ji})): 1 \le i \le n_j \}$; that is, the adversary obtains the set of all TCP segments and VSR tags for the transfer of file $fid$ for a particular TCP session $j$.   The adversary is allowed to make {\bf \em nested} queries of segment queries, tag access queries and proof verification queries in each playback query.  Denote the set of queried file identifiers in Query 1 phase by ${\cal F}_1$.

\item[Challenge.]
Once the adversary decides that the first query phase is over, it selects a file identifier $fid$ and a default segment size $l$ to ask the challenger for a challenge TCP session.  There is no constraint imposed on the choice of file.  Denote the file content of $fid$ by $M$.  The challenger runs standard TCP initiation handshake with the adversary and determines TCP headers to be used.  The challenger generates a new secret key $sk$ and breaks down $M$ to fill say $n$ TCP segments.  The adversary is challenged with the task of creating a correct proof for TCP segments he has not queried in the second query phase.

\item[Query 2.]
The adversary is allowed to make more queries to all oracles as previously done in Query 1 phase except proof verification queries on the challenged TCP session unless the reply result is $1$.\footnote{A TCP sender should terminate a TCP session when an incorrect proof is received from the receiver.}  Playback queries on the challenge $fid$ is even allowed; note that, for playback queries, even the same file identifier (as chosen in the challenge TCP session) is requested, a different secret key and random coin will be used for the new query session (as described in the playback oracle). Denote the set of indices for ${\cal O}_{SQ}$ and ${\cal O}_{TA}$ queries {\em on the challenge TCP session} by ${\cal I}_{SQ}$ and ${\cal I}_{TA}$ respectively.  Let the set of queried file identifiers $fid_j$ in Query 2 phase be ${\cal F}_2$.

\item[Guess.]
Finally, the adversary outputs a guess $(I, p(I))$ for some $I \subseteq [1,n]$ such that $I \backslash {\cal I}_{SQ} \neq \phi$ (where $\phi$ is the empty set).

\item[Result.]
The adversary wins the game if $(I, p(I))$ passes the verification test $\textsf{PV}$.  The advantage $Adv_{\cal A}$ of the adversary ${\cal A}$ is defined as the probability of winning the game.
\end{description}

\subsubsection{Security against Known File Content Only Attacks ($\textsf{KFC}$-secure).}

Security against known file content attacks for the VSR protocol is defined by the same game in Definition~\ref{def:kfc-ta-secure} except that no tag access (${\cal O}_{TA}$) queries can be made on the challenge TCP session in Query 2 phase.  That is, ${\cal I}_{TA} = \phi$.

\subsubsection{Security against Tag Access Only  Attacks ($\textsf{TA}$-secure).}
Security against tag access only attacks for the VSR protocol is again defined by the same game in Definition~\ref{def:kfc-ta-secure} except two modifications as depicted below.
\begin{enumerate}
\item In the Challenge phase, the adversary has to choose a challenge $fid \not \in {\cal F}_1$, that is, $fid$ has to be a new file identifier not used as input for the playback queries in the Query 1 phase.

\item In the Query 2 phase, the adversary can only make playback queries on $fid_j \neq fid$ where $fid$ is the file identifier for the challenge TCP session.  That is, $fid \not \in {\cal F}_2$.
\end{enumerate}

\medskip
Informally speaking, a $\textsf{KFC-TA}$-secure VSR protocol construction would be able to catch a misbehaving TCP receiver who has full knowledge of the file content in consideration and has installed agents to help him obtain VSR tags through some side channel.  A $\textsf{KFC}$-secure or $\textsf{TA}$-secure VSR protocol construction can only withhold weaker attacks: A $\textsf{KFC}$-secure VSR protocol construction can catch a lying TCP receiver who has full knowledge of the file content but is unable to grab VSR tags through a side channel, whereas, a $\textsf{TA}$-secure VSR protocol construction can withhold a misbehaving TCP receiver who can grab VSR tags through a side channel but has no knowledge about the file content being transferred.

\section{\label{sect:construction} A Construction of Verifiable Segment Receipt Protocol}

A construction of the VSR protocol for use in TCP is given, based on an aggregate authenticator introduced in \cite{CCMT08}.  We call this protocol Verifiable Segment Receipt protocol with Aggregate Authenticator (VSR-AA).  An aggregate authenticator assures the integrity of messages while allowing (additive) aggregation to be performed separately on messages and tags.

Besides being proven secure, the main advantage of the VSR-AA protocol over the random number approach in \cite{SCWA99} is that the TCP sender in the VSR-AA protocol does not have to store any tags while the corresponding TCP segments are in transit to the receiver. In fact, the sender does not even need to store the TCP segment itself but is still able to verify (without errors) the correctness of a proof returned by the TCP receiver; that is, any incorrect proofs by a misbehaving receiver would be caught with high probability.   This advantage renders the VSR-AA compatible with the TCP-SACK option \cite{rfc2018}, whereas, to use the TCP-SACK option, the approach in \cite{SCWA99} cannot use cumulative sum of random numbers to reduce storage but has to store each random number separately.

                        \vspace{0.1cm}

                        {
                        }

                        {
                        }

The VSR protocol is constructed based on an (additive) aggregate authenticator \cite{CCMT08} which allows aggregation on messages and aggregation on tags.
In turn, the aggregate authenticator in \cite{CCMT08} is constructed from a pseudorandom function (PRF) \cite{GGM86,GGM84}.  For an in-depth treatment of PRFs, we refer to \cite{gold98}. In our context, a PRF is needed to derive secret keys for different TCP segments from the TCP session keys (freshly picked for each new TCP session).

Let $F = \{F_{\lambda}\}_{\lambda \in \mathbb{N}}$ be a PRF
family where $F_{\lambda} = \{ f_s: \{0,1\}^{\lambda} \rightarrow
\{0,1\}^{\lambda} \}_{s \in \{0, 1 \}^{\lambda}}$ is a collection of
functions\footnote{While we denote both the input and the secret key of $f$ with the same length, they need not have the same length in PRFs, namely, we can have a PRF of the form: $f: \{0,1\}^{\lambda} \times \{0,1\}^{l_i} \rightarrow \{0,1\}^{l_o}$ where the key length is $\lambda$, input length is $l_i$ and output length is $l_o$.} indexed by key $s \in \{0,1\}^{\lambda}$. Informally,
given a function $f_s$ from a PRF family with an unknown key $s$,
any PPT distinguishing procedure allowed to get the values of
$f_s(\cdot)$ at (polynomially many) arguments of its choice should
be unable to distinguish (with non-negligible advantage in
$\lambda$) whether the answer of a new query is supplied by $f_s$ or
randomly chosen from $\{0,1\}^{\lambda}$.  The VSR-AA protocol is both $\textsf{KFC}$-secure and $\textsf{TA}$-secure if the underlying key derivation function $f$ is a pseudorandom function.

The arithmetics for the VSR-AA protocol is done in some finite field $\mathbb{Z}_p$ where $p$ is some large prime (at least $\lambda$-bit long).  Alternatively, the extension field of $\mathbb{Z}_2$ can be used; the advantage is that extension fields of $\mathbb{Z}_2$ can usually result in efficient logic circuits and efficient computation procedures on most platforms.

A hash function $h: \{0, 1\}^* \rightarrow \mathbb{Z}_p$ is used in the VSR-AA protocol.  The hash function used in VSR-AA is merely for the mapping purpose (mapping an arbitrary binary string to an element in $\mathbb{Z}_p$) and no security requirement is thus needed.  In case the segment size is a multiple of $|p|$ (the length of $p$), $h$ can simply be implemented by breaking down the input segment into units of length $|p|$ (each of which can be represented as an element in $\mathbb{Z}_p$) and summing up all these units in $\mathbb{Z}_p$ and treating the sum as the output of $h$.

\medskip
\noindent The VSR-AA protocol runs as follows.
\medskip

\begin{table}[phtb]
  \begin{center}
      \begin{tabular}{p{5.5in}} \hline
                  The VSR Protocol based on an Aggregate Authenticator (VSR-AA) \\
                        \hline
                            \\
                        Assume the file being transferred can be broken down to fill in a total of $n$ TCP segments indexed by $i$.
                        \vspace{0.1cm}

                        {\it Setup ($\textsf{setup}$):}\\
                        {
                        For a security parameter $\lambda$, choose a large prime $p$ where $p$ has to be at least $\lambda$ bits long.  Choose a pseudorandom function $f: \{0,1\}^{\lambda} \times \{0,1\}^* \rightarrow \mathbb{Z}_p$ (with key length $\lambda$).  Let $h: \{0,1\}^* \rightarrow \mathbb{Z}_p$ be a chosen length-matching hash function.
                        }
                        \vspace{0.1cm}

                        {\it Key Generation ($\textsf{KG}$) by $S$:}\\
                        {
                        Randomly choose two keys $K \in \mathbb{Z}_p^*$ and $K' \in \mathbb{Z}_p^*$.  $K'$ is the session key used to derive segment keys while $K$ is used as a long term segment key for a TCP session.
                        }
                        \vspace{0.1cm}

                        {\it Tag Generation ($\textsf{TG}$) by $S$:}\\
                        {
                        For each TCP segment $s_i$ with index $i$,
                        \begin{enumerate}
                        \item  retrieve the session keys $K, K'$;
                        \item  generate a new segment key $k_i$ by computing $k_i = f_{K'}(i)$;
                        \item  compute the tag $t_i = K \cdot h(s_i) + k_i \text{ mod } p$;
                        \item  $(s_i,t_i)$ is sent to the receiver.
                        \end{enumerate}
                        }
                        \vspace{0.1cm}

                        {\it Proof Generation ($\textsf{PG}$) by $R$:}\\
                        {
                        To generate a proof for the receipt of all segments with indices in $I$,
                        \begin{enumerate}
                        \item for each $i \in I$, retrieve $(s_i, t_i)$;
                        \item compute $x = \sum_{i \in I} h(s_i) \text{ mod } p$;
                        \item compute $y = \sum_{i \in I} t_i \text{ mod } p$;
                        \item output $(x,y)$ as the proof $p(I)$.
                        \end{enumerate}
                        }
                        \vspace{0.1cm}

                        {\it Proof Verification ($\textsf{PV}$) by $S$:}\\
                        {
                        Given a proof $(I, x' y')$, the verification is performed as follows.
                        \begin{enumerate}
                        \item Retrieve the session keys $K, K'$;
                        \item Compute $K_I = \sum_{i \in I} f_{K'}(i) \text{ mod } p$.
                        \item Check $y' \isequal K_I + K \cdot x' \text{ mod } p$.  If yes, return $1$, otherwise, return $0$.
                       \end{enumerate}
                            }
                     \\ \hline
                \end{tabular}
  \end{center}
\caption{Operation of VSR-AA}
\label{table:vsr-aa}
\vskip -0.3cm
\end{table}

\pagebreak


\noindent The correctness of the VSR-AA protocol can be easily verified as follows:

\noindent For any $I \in [1,n]$,
\begin{equation}
\label{eq:tag-gen}
\begin{array}{l}
x = \sum_{i \in I} h(s_i) \text{ mod } p \quad \text{and} \quad y = \sum_{i \in I} t_i \text{ mod } p, \qquad \text{ where } t_i = K \cdot h(s_i) + k_i \text{ mod } p.
\end{array}
\end{equation}

\noindent Substituting $t_i$ into $y$, we have
\[
\begin{array}{lcl}
y & = & \sum_{i \in I} t_i \text{ mod } p\\
  & = & \sum_{i \in I} (K \cdot h(s_i) + k_i) \text{ mod } p\\
  & = & K \cdot \sum_{i \in I} h(s_i) \text{ mod } p + \sum_{i \in I} k_i \text{ mod } p \\
  & = & K \cdot x + \sum_{i \in I} k_i \text{ mod } p
\end{array}
\]

\noindent Substituting $k_i = f_{K'}(i)$, we have $\sum_{i \in I} k_i \text{ mod } p = \sum_{i \in I} f_{K'}(i) \text{ mod } p = K_I$.  Hence,
\begin{equation}
\label{eq:check}
y = K \cdot x + K_I \text{ mod } p.
\end{equation}
That is, for any correctly generated $x$ and $y$ according to Equation (\ref{eq:tag-gen}), they would fulfill the check equation, Equation (\ref{eq:check}), and pass the proof verification.

The security of the VSR-AA protocol is based on the difficulty in determining the actual coefficients used in an under-determined equations.  More specifically, given $x$ and $y$ where $y = K' + K \cdot x$ for fixed, secret $K$ and $K'$, determine the actual $K$ and $K'$ used in forming $x$ and $y$.  If $K$ and $K'$ are randomly picked, it can be shown that $(x, y)$ does not give sufficient information to determine $K, K'$.  For arithmetics in $\mathbb{Z}_p$, there are $p$ possible 2-tuples of $(k, k')$ which can lead to the given $(x,y)$; only one out of $p$ is the actual pair $(K, K')$.  Hence, for a large enough prime, there is a negligibly small probability $1/p$ to guess $(K, K')$ correctly.  Nevertheless, this problem can be solved with ease if an adversary breaking the VSR-AA protocol exists.  The idea is as follows.

Suppose we are given a pair $(x,y)$ and asked to determine the actual $K, K'$ used.  Without loss of generality, assume a single segment in the following discussion.  This $(x, y)$ can be treated as the expected VSR proof $(h(s_i),t_i)$ for a segment-tag pair $(s_i, t_i)$ sent out to the receiver.\footnote{Note there is no need to find $s_i$ such that $h(s_i) =x$ although it is possible.  $(s_i,t_i)$ will not be sent to the receiver who tries to lie about the receipt of it. }  In order to lie about the receipt of a pair $(s_i, t_i)$ which is actually lost and to convince the sender, the receiver needs to determine a 2-tuple $(x',y')$ to fulfil the constraint equation:
\begin{equation}\label{eq-vsr-aa}
y' = K' + K \cdot x'
\end{equation}
where $K$ and $K'$ are unknown.

Note also that the receiver has no knowledge about $(x,y)$.
There are $p$ possible 2-tuples of $(x', y')$ fulfilling equation (\ref{eq-vsr-aa}), one of which is $(x,y)$.  Suppose the tuple $(x', y')$ fulfills equation (\ref{eq-vsr-aa}).  The probability that $(x',y') \neq (x,y)$ is $\frac{p-1}{p} \approx 1$ (for large $p$).  In other words, we have two independent equations: (1) $y' = K' + K \cdot x'$ ; (2) $y = K' + K \cdot x$ to solve $K$ and $K'$ which is easy.  Consequently, we can say the VSR-AA protocol is secure.  A more rigorous argument is given in the proof.

\smallskip
\noindent The security of the VSR-AA protocol is summarized by the following theorem, with the proof in Appendix \ref{appendix:proof-vsr-aa}.
\begin{theorem}
Assuming $f$ is a pseudorandom function, the VSR-AA protocol is both $\textsf{KFC}$-secure and $\textsf{TA}$-secure.
\end{theorem}

\section{\label{sect:compute} Practical Considerations}

\subsection{\label{ssect:sec-choice} Choices of PRF}
For the VSH-AA protocol to achieve $2^{80}$ security, $p$ and hence $\lambda$ need to be at least $80$ bits.  That is, each tag is $80$-bit long and each proof is $160$-bit long.\footnote{The overhead of $I$ is inherently included in the TCP acknowledgement and does not need to be counted as additional communication overhead.}
Most provably secure PRFs such as \cite{NRR02} are based on the hardness of certain number-theoretic problems.  However, such constructions are usually computationally expensive. Instead, key derivation in practice is often based on functions with conjectured or assumed pseudorandomness, i.e., it is inherently assumed in the construction rather than proven to follow from the hardness of a specific computational problem.  One common example is the use of cryptographic hash functions for key derivation, such as \cite{SWCT01}.  Some well-known primitives, such as HMAC \cite{BCK96} and OMAC \cite{IK03} (conjectured PRFs), are based on assumed pseudorandomness. HMAC assumes that the underlying compression function of the hash function in use is a PRF, while OMAC assumes the underlying block cipher is a pseudorandom permutation.

The VSR-AA protocol in this paper does not impose any restriction on the type of PRFs to be used. The security guarantee provided by the proposed construction holds as long as the underlying PRF has the property of pseudorandomness or indistinguishability.  We note that the aforementioned pseudorandomness property is also a basic requirement for the hash function used for key derivation purposes \cite{SWCT01,BCK96}, for instance, in the well-known IPSec standard.

\subsection{\label{ssect:overhead} Computational Overhead}

Assume $\lambda$ is the used security parameter, that is, $p$ in the VSR-AA protocol construction is $\lambda$ bits long.
Let $t_{add}$ and $t_{multi}$ denote the respective costs of performing a $\lambda$-bit addition and multiplication in $\mathbb{Z}_p$.  Note that $t_{add} \sim O(\lambda)$ while $t_{multi} \sim O(\lambda^2)$ (See \cite{MOV96} Chapter 2).  Let $t_{prf}$ denote the cost of evaluating the PRF plus the hash function in the VSR-AA protocol with security parameter $\lambda$.  Note that the cost of evaluating $h$ in the VSR-AA protocol is negligible compared to that of evaluating the PRF.
Let $S$ and $R$ denote the sender and receiver respectively.  Let $w$ be the maximum window size (in number of segments) for a TCP connection.  The overhead of the VSR-AA protocols is summarized as follows.

\begin{table*}[tbph]
\footnotesize
\begin{center}
        \begin{tabular}{|l||c|} \hline
              & VSR-AA \\ \hline \hline
        Storage (S)  & $0$\\
        Computation Overhead (Tag Generation) per Segment (S) & $t_{prf} + t_{add} + t_{multi}$ \\
        Computation Overhead (Tag Verification) per Window (S) & $w \cdot (t_{prf} + t_{add}) + t_{multi}$ \\
        Communication Overhead per Segment Sent (S) & $\lambda$ \\
        \hline
        Computation Overhead (Proof Generation) per Window (R) & $2 (w-1) \cdot t_{add}$ \\
        Communication Overhead per Window (R) & $2 \lambda$ \\
        \hline
        \end{tabular}
\caption{\label{table:vsr-overhead}
Overhead of the VSR-AA protocol (assuming $|I| = w$).  The secret key storage is not counted.}
\end{center}
\vskip -0.5cm
\end{table*}

\subsubsection{Experimental Results on Computational Overhead}
Some empirical studies have been done on an Intel Pentium 4 3.00 GHz CPU, based on a secret key size of 80 bits and $|p| = 80$ bits.  The tag generation and tag verification for each TCP segment (with a payload size of 536 bytes) would take roughly 0.86 and 0.85 microseconds respectively.  With these figures, we believe the proposed VSR protocol is efficient enough to be used with TCP.  More importantly, the proposed protocol is provably secure, meaning the only way to break its security is to break the security property of the underlying pseudorandom function.

\subsection{\label{ssect:comm-overhead} An Estimate of the Communication Overhead}
The overhead size of the VSR-AA protocol in sending each TCP segment is $\lambda$ bits.  For $2^{80}$ security, it translates to a tag size of $80$ bits or 10 bytes.  As specified in RFC 879 \cite{rfc879}, the MTU (Maximum Transmission Unit) of an IP (Internet Protocol) datagram is 576 bytes.  This corresponds to a TCP payload size of 536 bytes (after subtracting the 20 bytes IP header and 20 bytes TCP header) during bulk data transfer, excluding the $SYN$-$ACK$ handshake messages.  Since MTU is adopted in the bulk data transfer phase\footnote{Note that TCP congestion control only adapts the congestion window size but not the TCP segment size.  MTU refers to the maximum allowed IP datagram size if no fragmentation is necessary.  Under normal circumstance, packing data to fit MTU is sought in almost all implementations.  Datagrams greater than MTU are preferred in some implementations, even though fragmentation and defragmentation are needed.} by default in almost all TCP implementations, and the handshake and other signalling messages only make up a constant overhead consuming a very small portion of the bandwidth, the payload size is asymptotically close to 536 out of 576 bytes for large data files to be transferred.   As a result, the communication overhead for embedding VSR tags for the VSR-AA protocol can roughly be estimated as $10/536 \approx 1.86\%$.  In a typical TCP session, most implementations would use MTU for data transfer, which seems to be the default.  Roughly, the average overhead would be about 2\% in most cases.

\subsection{\label{ssect:backward-compatibility} Backward Compatibility}
Clearly, the victims (which are normally TCP senders) have the biggest incentive to install the new TCP code running the VSR protocol.  The most common victims of the DoS attack in question are web servers, which have to accept connections from any client machines.  Normally, the owners of such machines would promptly upgrade or patch the machines' software, in particular, when a new vulnerability is discovered.  It is thus fair to say, initially, the TCP sender of most TCP sessions would have the capability of running the VSR protocol, and backward compatibility would only be an issue for TCP receivers.  For TCP receivers, some may be able to run the VSR protocol while others may not.

Although the proposed VSR protocol requires code modification at both TCP senders and receivers, the new TCP version running the VSR protocol could still be rolled out incrementally, with both new and old TCP receivers co-existing initially.  More specifically, the VSR protocol could be implemented as a TCP option \cite{rfc793}, with the VSR tag embedded in the TCP option field of a TCP segment; in other words, an old TCP receiver has the choice to opt out running the VSR protocol and need not generate a VSR proof.  In such cases, the new TCP sender could arbitrarily set and impose an upper limit on the transmission window size of all the TCP sessions established with old TCP receivers, and the limit could be a fraction of the maximum allowable size as permitted by the TCP congestion control algorithm.  In other words, the new TCP sender limits sending traffic to an old TCP receiver by a rate up to a fraction of the maximum transmission rate (as allowed by the TCP congestion control algorithm), while sending at full transmission rate (allowed by TCP congestion control algorithm and the VSR protocol) to a new TCP receiver.  The rationale is that, for a less trustworthy machine (running the old TCP receiver code), less privileges (in terms of bandwidth provided for its traffic) should be granted.  Such penalty is justified as machines running the new TCP receiver code have to prove that they are giving correct feedback to the sender, and hence are comparatively more trustworthy and deserve a full rate of receiving packets.  Hopefully, this would motivate the old TCP receiver machines to upgrade their codes to run the VSR protocol.

\subsection{Other Potential Applications}
\label{ssect:other-applications}
While the VSR protocol is designed based on an end-to-end flow control mechanism provided by TCP, its function and objective are to ensure correctness of feedback (about congestion condition downstream) returned from the receiver.  In other words, the VSR protocol can be applied to other network-layer rate-limiting mechanisms requiring security protection or correctness for feedback such as \cite{BJCSSK05,LYX10}.  A main advantage of the VSR protocol to such application is that the feedback could potentially be aggregated and have a smaller communication overhead.

\section{Related Work}
\label{sect:related-work}
Almost all DoS attacks recorded in the literature are carried out by compromised senders, sending high-volume traffic to the victim (receiver) aiming to cause network congestion and deny legitimate communication as a result.  Consequently, except the work based on in-network fair queuing with explicit router congestion notification \cite{BJCSSK05,LYL08,LYX10},\footnote{Collusion between sender and receiver is assumed possible in these approaches.} the majority of approaches to defend DoS attacks assume that the sender is malicious while the receiver is a victim.  Some of these defence mechanisms even trust the receiver to an extent that the defence or traffic policing has to be triggered on request by the receiver, for instance, pushback \cite{MBFIPS02}, AIP \cite{ABFKMS08} and AITF \cite{AC09}.  Clearly, there is no incentive for the misbehaving TCP receiver to push back or to contact the sender or its access router when launching an attack; in effect, these defence mechanisms could provide no protection to a DoS attack initiated by a misbehaving TCP receiver.  The crux of the misbehaving TCP receiver attack is that an additional level of indirection may switch the role of the attacker and victim, thus bypassing many of the existing defence mechanisms assuming the receiver as the victim.  In the misbehaving TCP receiver attack, by returning incorrect feedback, the attacker (as a receiver) deceives normally trustworthy senders (which could be highly trusted machines such as the web server for a highly reputable website, for instance, those owned by the national government) to flood the network.  A sophisticated attacker could even hide such an attack from undesired traffic identification program \cite{DDOS10} by selectively manipulating a larger number of innocent senders in such a way that the traffic of each individual sender may still comply to the acceptable bound most of the time.  Only per-receiver fair queuing schemes at the congested router such as \cite{YWA08,LYL08} can effectively police such high-volume aggregate traffic.  However, the feasibility of deployment of such schemes is still uncertain due to the large overhead involved in storing flow states.  Liu et. al. \cite{LYX10} recently proposed a rate-limiting architecture called NetFence, reducing the overhead for storing flow states.  In NetFence, traffic policing is performed by access routers based on unforgeable feedback from congested routers.  NetFence is based on network-layer policing and router feedback, whereas, the VSR protocol is based on sender flow control and end-to-end receiver feedback.  Hence, NetFence requires changes to the router code (as well as the receiver code for returning feedback messages) while the VSR protocol only requires end-host (sender and receiver) modification.  Besides, compared to the tags in the VSR protocol, feedback messages in NetFence are in the form of a relatively longer message secured by a message authentication code (MAC).  The MAC is slightly longer than the VSR tag, but the feedback message is much longer than the MAC or VSR tag.  On the other hand, the VSR protocol could possibly be applied to the NetFence architecture, with the advantage that only access routers need to be changed and no delivery of secret key from access routers to downstream routers is necessary.\footnote{The access routers would perform both VSR tag generation and verification.}

Based on approaches proposed in the literature, Clark gave a comprehensive list of design principles to defend DoS attacks in the Future Internet \cite{Clark09}.  Although some of the suggestions can reduce the attack's impact (mainly by degrading the botnet), the misbehaving TCP receiver DoS attack cannot be completely precluded with such principles.  In contrast, the proposed VSR protocol can completely eliminate the misbehaving TCP receiver attack.  Clark's suggestions include distinguishing between trusted and untrusted machines, degrading botnets, detecting and flagging infested attack machines, diffusing attacks, and rate-limiting traffic.
Of course, techniques for degrading botnets, detecting infested attack machines and diffusing attacks could reduce the impact of the TCP receiver attack.  However, it is still possible for a deliberate attacker to bypass these mechanisms, and the impact of the attack carried out by a single TCP receiver is still significant in some scenarios.  Besides, as Clark mentioned, it is important to assume that compromised nodes always exist and to design an architecture resilient to the presence of compromised nodes \cite{Clark09}.  The VSR protocol assumes the presence of compromised receiver and forces them to tell the truth regarding the network congestion condition.

In general, it is difficult to distinguish between trustworthy and untrusted machines, in particular, when the receiver is malicious as in the TCP receiver attack.  First, the senders are trustworthy in the TCP receiver attack.  Moreover, they are usually web servers which have to serve content to anyone.  That means the senders cannot discriminate connection requests from a potential attacker.  Furthermore, even with cryptography, it is difficult to guarantee that a host is uncompromised.  Detecting a malicious machine is uneasy, particularly that a machine could potentially be compromised after a connection is established.  Although CAPTCHA may be used to detect botnets, it still has its limitations such as accessibility problems and potential attacks.  In fact, infrequent human intervention would possibly thwart CAPTCHA and make the TCP receiver attack successful.  Second, allowing nodes to control who can send to them, the so-called "wish list" approach, is only effective to malicious senders, but does not protect against a misbehaving receiver.  Similarly, source authentication, such as AIP \cite{ABFKMS08} which equips all packets with self-certifying unspoofable addresses and \cite{LLYW08}, can only guard against a malicious sender, but not a misbehaving receiver.  Again, network capability based approaches such as \cite{PWSPMH07}, which enable a receiver to deny by default all traffic and explicitly accept traffic from identified legitimate and trustworthy sources, would not work either, with respect to the misbehaving TCP receiver attack.

Rate-limiting is the most effective defence against DoS attacks: attacking machines may be rate-limited in ways that allow legitimate users to continue preferentially if a strong framework for dealing with congestion is in place.  All rate-limiting architectures would rely on a certain triggering or detection mechanism.  If such a mechanism is initiated by the receiver such as \cite{MBFIPS02,ABFKMS08,AC09}, rate-limiting would be crippled by the TCP receiver attack.  Design based on triggering from congested routers such as \cite{BJCSSK05,LYL08,LYX10} or which performs rate-limiting through per-receiver fair queueing can effectively mitigate the TCP receiver attack, with various complexity in modification needed.  Network-layer congestion control is an important foundation for these mechanisms.  On the contrary, the VSR protocol is based on end-to-end rate-limiting, with round-trip congestion feedback and sender-controlled flow control.  The VSR protocol provides guarantee that the feedback from the receiver truly reflects the congestion status as can be inferred from end-to-end packet loss.

\section{Conclusions}
\label{sect:conclusions}
We believe that any potential attacks in a core protocol of the Internet protocol suite would jeopardize the resilience of the current day Internet, and the impact of misbehaving receiver could possibly be one of them.  We address the problem of misbehaving TCP receivers which lie about the receipt of TCP segments not actually received in order to cause congestion collapse in the Internet.  Our solutions are based on cryptographic tools, namely, an aggregate authenticator constructed based on a pseudorandom function.  The protocols proposed are provably secure and reasonably efficient for incorporation into the TCP protocol.  We demonstrate that the proposed protocol is sufficiently simple and efficient for incorporating into the TCP protocol.  This is the first light-weight provably secure solution to the problem of misbehaving TCP receiver.  Besides, the protocol could be applied to other network-layer rate-limiting architectures.

\vskip 2cm
\noindent {\bf Reference}
\bibliographystyle{plain}
\bibliography{./tcpsec,../rfc,../thesis,../key,../mybook}

\vskip 2cm
\appendix
\centerline{{\bf \sixteen Appendices \normalsize}}

\section{Proof of Security of VSR-AA}
\label{appendix:proof-vsr-aa}

\begin{proof}
Assume the PRF has some indistinguishability property as usual. We prove by contradiction, showing that a PPT adversary which can forge a valid pair $(x', y')$ (recall that $(x', y') = (x,y)$ with probability $1/p$) can also break the indistinguishability property of the underlying PRF.  We show the reduction\footnote{The reduction of the problem of breaking the indistinguishability of the PRF to the problem of forging a valid $(x', y')$ pair to pass the verification while $(x,y)$ is not received.} in two steps: first, we show that a forging algorithm to find $(x', y')$ can be used as a sub-routine to solve a newly defined problem called ``Under-determined Equation Set with Pseudorandom Unknowns (UESPU)"; then we show that the UESPU problem is computationally hard if the underlying PRF has the usual indistinguishability property.  The UESPU problem is defined as follows:

\medskip

\noindent {\bf \em Under-determined Equation Set with Pseudorandom Unknowns (UESPU) Problem} --- Suppose $K, K'$ are independent random seeds. Denote $K_I'$ as the sum $\sum_{i \in I} f_{K'}(i)$   Given a 2-tuple $(x, y)$ where $y = K_I' + K \cdot x$, find $(K, K_I')$ while allowed to evaluate the PRF at any input $i \not \in I$.

\medskip

\noindent {\bf Solving the UESPU problem using a forger of $(x',
y')$.}

Suppose there exists a PPT adversary ${\cal A}$ which can forge a valid pair $(x', y')$ to pass the VSR-AA proof verification test with probability $p_f$.  Using ${\cal A}$ as a subroutine, we can construct another algorithm ${\cal A}'$ to find $(K, K_I')$ from $(x, y)$ with probability $\frac{p-1}{p}p_f$ in any instance of the UESPU problem. Note that ${\cal A}'$ should be able to answer ${\cal O}_{TA}$ and ${\cal O}_{SQ}$ queries from ${\cal A}$ for any $i' \not \in I$ by passing the queries to its own challenger.  Note that ${\cal A}'$ can answer all ${\cal O}_{PB}$ and ${\cal O}_{PV}$ queries without external help.  For playback queries, ${\cal A}'$ can simply pick new seed keys to run a new session.

The construction of ${\cal A}'$ is as follows: Give ${\cal A}$ the pair $(x, y)$. When ${\cal A}$ returns a pair $(x', y') \neq (x, y)$, we can determine $K, K_I'$ from the resulting set of equations. The explanation is as follows:

\noindent Note that
$
\begin{array}{lcl}
y & = & K_I' + K \cdot x.
\end{array}
$
So we have one equations and 2 unknowns. If $(x', y')$ is a valid
forgery, then it must satisfy the following two equations (with the
same $K, K_I'$) in order to pass the verification test:
\[
\begin{array}{lcl}
y' & = & K_I' + K \cdot x'.
\end{array}
\]

The pair $(x', y')$ adds in one new, independent equation.  Since $(x', y') \neq (x, y)$ with probability $\frac{p-1}{p}$, it can be assured that the two equations are independent with high probability for large $p$.  Hence, there are two independent equations and two unknowns in total and it should be easy to solve for $K, K_I'$ (a contradiction to the UESPU assumption). The probability of solving the problem in the UESPU assumption is hence $\frac{p-1}{p}p_f$.

\medskip

\noindent {\bf A distinguisher for the PRF using
an algorithm which solves the UESPU problem.}

The UESPU problem is hard if $K, K_I'$ are generated by a PRF.  There are one equation and two unknowns which cannot be uniquely determined. It could be shown that if there exists an algorithm ${\cal A}'$ solving in poly-time $K$ and $K_I'$ from $x$ and $y$, then the indistinguishability property of the underlying PRF is broken.

The idea is as follows: assume the seed key for generating $K_I'$ is unknown and the key $K$ is known.  When a challenge $K_I'$ is received, we have to determine whether it is randomly
picked from a uniform distribution or generated by the PRF with an unknown seed key. We compute $y = K_I' + K \cdot x$ to ${\cal A}'$. If the solution from ${\cal A}'$ does not match the
generated $K_I'$, we reply that $K_I'$ is randomly picked, otherwise, it is generated from the PRF. If ${\cal A}'$ has non-negligible probability of breaking the UESPU assumption, the above construction would also has a non-negligible advantage of breaking the indistinguishability property of the underlying PRF. Note that all queries from ${\cal A}'$ could be answered by sending queries to the challenger and running the PRF with the known key.
\end{proof}

\end{document}